\documentclass[a4paper,11pt]{article}
\usepackage{algorithm}
\usepackage{algpseudocode}
\usepackage[margin=1in]{geometry}
\usepackage{enumerate}
\usepackage{amsmath,amsfonts,amssymb,amsthm}
\usepackage{mathrsfs}
\usepackage{graphicx}
\usepackage{xspace}
\usepackage{enumerate}

\usepackage{thm-restate}

\newtheorem{theorem}{Theorem}
\newtheorem{lemma}{Lemma}

\newtheorem{corollary}[lemma]{Corollary}
\newtheorem{observation}[lemma]{Observation}

\newcommand{\myurl}[2]{\url{#1}\xspace}
\newcommand{\useless}[1]{}
\usepackage{afterpage}

\newcommand{\calO}{\mathcal{O}\xspace}

\newcommand{\bbX}{\mathbb{X}\xspace}

\begin{document}
\title{A Brief Note on Single Source Fault Tolerant Reachability}
\author{
    Daniel Lokshtanov\thanks{University of California Santa Barbara, USA. \texttt{daniello@ii.uib.no}
    } \and 
    Pranabendu Misra\thanks{Max-Planck Institute for Informatics, Saarbrucken, Germany. \texttt{pranabendu.misra@ii.uib.no}
    } \and
    Saket Saurabh\thanks{Institute of Mathematical Sciences, Chennai, India. \texttt{saket@imsc.res.in} 
    } \and
    Meirav Zehavi\thanks{Ben-Gurion University, Beersheba, Israel. \texttt{meiravze@bgu.ac.il}
    } 
}

%
%
%
%

\date{}

\maketitle
    \begin{abstract}
    Let $G$ be a directed graph with $n$ vertices and $m$ edges, and let $s \in V(G)$ be a designated source vertex. We consider the problem of single source reachability (SSR) from $s$ in presence of failures of edges (or vertices). Formally, a spanning subgraph $H$ of $G$ is a {\em $k$-Fault Tolerant Reachability Subgraph ($k$-FTRS)} if it has the following property. For any set $F$ of at most $k$ edges (or vertices) in $G$, and for any vertex $v\in V(G)$, the vertex $v$ is reachable from $s$ in $G-F$ if and only if it is reachable from $s$ in $H - F$. 
    Baswana et.al. [STOC 2016, SICOMP 2018] showed that in the setting above, for any positive integer $k$, we can compute a $k$-FTRS with $2^k n$ edges.
    In this paper, we give a much simpler algorithm for computing a $k$-FTRS, and  observe that it extends to higher connectivity as well. Our results follow from a simple application of \emph{important separators}, a well known technique in Parameterized Complexity.
\end{abstract}
    

\section{Introduction}\label{sec:intro}

Fault tolerant data structures aim to capture properties of real world networks, which are often prone to a small number of failures. Such data structures allow us to test various properties of the network after failures have occurred, and the repairs are awaited. The problem is modeled as a directed graph (digraph) $G$ where a small number of edges (or vertices) have failed, and a parameter $k$ is used as a bound on the maximum number of failures that may occur at a time. 
In this paper, we consider the problem of deciding the reachability of all vertices $v \in V(G)$ from a designated source vertex $s \in V(G)$ upon the failure of any $k$ edges (or vertices) in the input graph $G$. Specifically, our objective is to construct a sparse spanning subgraph $H$ of $G$ that preserves all reachability relationships from the source vertex $s$ upon the failure of any $k$ edges (or vertices) in $G$. More formally, we seek a spanning subgraph $H$ of $G$ with the following property: For any set $F$ of at most $k$ edges (or vertices) in $G$, and for any vertex $v\in V(G)$, there is a path from $s$ to $v$ in $G-F$ if and only if there is a path from $s$ to $v$ in $H - F$.
Such a graph $H$ is called a \emph{$k$-Fault Tolerant Reachability Subgraph ($k$-FTRS)}.
Observe that, beyond the question of deciding the reachability of a vertex $v$, the graph $H$ may also be used to find an alternate route from $s$ to $v$, if one exists, upon the failure of the edges in $F$. Now, the problem is formally defined as follows. 
Given as input a digraph $G$, a designated source vertex $s \in V(G)$ and an integer $k$, we must output a spanning subgraph $H$ of $G$ that is a $k$-FTRS.

Recently, Baswana et al.~\cite{BaswanaSICOMP18} presented an algorithm for computing a $k$-FTRS. Specifically, their algorithm runs in time $\calO(2^kmn)$ for a digraph $G$ of $n$ vertices and $m$ edges, and produces a $k$-FTRS where the in-degree of any vertex is upper bounded by $2^k$. Their algorithm is based on the notion of farthest min-cut that was introduced by Ford and Fulkerson~\cite{FordFulkerson}. They suggest that their methods may be of independent interest in other problems.
This is indeed so, for the notion of \emph{important separators}, which generalizes the notion of furthest cuts, is a well-known technique in Parameterized Complexity~\cite{PCBook}.

The notion of \emph{important separators} was introduced by Marx~\cite{MarxIPEC04} to give an FPT algorithm for {\sc Multiway Cut}. Subsequently, important separators and techniques based on them have been used to resolve the complexity of several important problems such as {\sc Directed Feedback Vertex Set}~\cite{ChenJACM08}, {\sc Multicut}~\cite{MarxSTOC11}, {\sc Directed Multiway Cut}~\cite{MarxSODA12}, {\sc Almost 2-SAT}~\cite{RazgonO09}, {\sc Parity Multiway Cut}~\cite{LokshtanovR12} and a linear-time FPT algorithm for {\sc Directed Feedback Vertex Set}~\cite{LokshtanovR018}. 
Informally speaking, important separators capture the entire collection of furthest cuts in a graph that have a bounded cardinality. We refer the reader to the textbook of Cygan et al.~\cite{PCBook} for an introduction to important separators, and more generally to the various tools and techniques in Parameterized Complexity.

Using the notion of important separators, we give a very simple and conceptually appealing algorithm for computing a $k$-FTRS. Indeed, we generalize the problem slightly. Given a digraph $G$, a designated source vertex $s \in V(G)$, an integer $\lambda$ and an integer $k$, we output a spanning subgraph $H$ of $G$ such that for any set $F$ of at most $k$ edges (or vertices) and for any vertex $v \in V(G)$, there are $\lambda$ edge-disjoint paths from $s$ to $v$ in $G - F$ if and only if there are $\lambda$ such paths in $H - F$. The graph $H$ is called a \emph{$(\lambda, k)$-Fault Tolerant Reachability Subgraph ($(\lambda, k)$-FTRS)}. As before, $H$ may be used for both testing the $\lambda$-connectivity of a vertex $v$ from $s$, as well as for obtaining an alternate collection of $\lambda$ edge-disjoint paths from $s$ to $v$, if they exist, after the failure of up to $k$ edges (or vertices). 
In particular, we obtain the following theorem.
\begin{theorem}\label{thm:k-ftss}
    There is a $(\lambda,k)$-FTRS where each vertex has in-degree at most $(k+\lambda) 4^{k+\lambda}$.
    Further, this graph can be constructed in time $\calO(4^{k+\lambda}(k+\lambda)^2(m+n)m)$.
\end{theorem}
Our bound and running time, obtained by a direct application of known results on important separators, are slightly worse than those given by Baswana et al.~\cite{BaswanaSICOMP18}. 
Our intent here is to provide a conceptual exposition of the results of Baswana et al~\cite{BaswanaSICOMP18} and further, introduce important separators as an algorithmic tool for problems on fault tolerant networks, and more generally for distributed~network~problems. Due to space constraints the discussion of other related works has been postponed to the appendix.

\paragraph{Related Work.} While a number of results are known about graph reachability for undirected graphs starting from the work of Nagamochi and Ibaraki~\cite{Nagamochi1992}, and more generally in the dynamic model\footnote{Here we have a sequence of edge insertions and deletions, as well as reachability queries, and we must efficiently update the data structure and answer the queries.}~\cite{DuanPettie16}, relatively little is known in the case of digraphs. 
Patrascu and Thorup~\cite{PatrascuThorup07} considered the problem of edge failures in undirected graphs. They constructed a data structure that processes a bath of $k$ edge failures in $\calO(k \log^2 n \log\log n)$ time, and then answers connectivity queries for any pair of vertices in $\calO(\log\log n)$ time. Subsequent results have led to a randomized data structure that uses almost-linear space (in the number of vertices) and correctly answers the queries with high probability~\cite{GibbKKT15}. For vertex failures in undirected graphs, Duan and Pettie~\cite{DuanPettie10} gave a data structure that can process a batch of $k$ vertex failures in $\calO(k^{2c+4}\log^2 n \log\log n)$ time and thereafter answer connectivity queries in $\calO(k)$ time. Here $c$ is a parameter of the algorithm offering a tradeoff between the running time and the space used. Recently, they improved this algorithm to process the $k$ failures in $\calO(k^3\log^3 n)$ time and with the same query time as before, while using $\calO(k m \log n)$ space~\cite{DuanPettie16,DuanPettie17}. We refer to~\cite{DuanPettie16} for further details.

In digraphs, other than~\cite{BaswanaSICOMP18}, an optimal oracle for dual fault tolerant reachability was proposed by Choudhary~\cite{ChoudharyICALP16}. Furthermore, fault tolerance in digraphs for shortest paths~\cite{DuanPettie09,BiloG0P16,GuptaK17,BodwinGPW17,BaswanaCHR18} and strongly connected components~\cite{GeorgiadisIP17,BaswanaCR17} were also studied. The various problems in the fault tolerant model for digraphs are subject to extensive ongoing research. We refer to~\cite{BaswanaSICOMP18,ChoudharyICALP16,DuanPettie16,BaswanaCR17,BaswanaCHR18} for a detailed discussion and further details.


\section{Preliminaries}
\label{sec:prelims}

Let $G$ be a digraph on $n$ vertices and $m$ edges.
For a subset of edges $X\subseteq E(G)$, we let $G-X$ denote the subgraph of $G$ with vertex set $V(G)$ and edge set $E(G) \setminus X$. We omit the braces when the set contains only a single edge, i.e.~$G -\{e\}$ is denoted by $G-e$. For an edge $e = (u,v)$, $u$ is called the tail of $e$ and $v$ is called the head of $v$. We denote these vertices by $\mathrm{tail}(e)$ and $\mathrm{head}(e)$, respectively.
For $R \subseteq V(G)$, $\delta^-(R)$ denotes the set of in-coming edges to $R$, i.e.~the set of edges of $G$ such that $\mathrm{tail}(e) \in V(G) \setminus R$ and $\mathrm{head}(e) \in R$. Similarly, $\delta^+(R)$ denotes the set of out-going edges from $R$.
Let $G$ be a digraph and $S$ and $T$ be two disjoint subsets of $V(G)$. A \emph{(directed) $(S,T)$-cut} is a subset $X$ of edges of $G$ such that there is no path from a vertex in $S$ to a vertex in $T$ in $G-X$. Any minimal $(S,T)$-cut $X$ can be expressed as $\delta^+(R)=X$ where $S \subseteq R \subseteq V(G) \setminus T$ is the set of vertices that are reachable from some vertex of $S$ in $G-X$. Hence, for any $(S,T)$-cut $X$, let $R_X$ denote the set of vertices that are reachable from $S$ in $G-X$. The set $R_X$ is called the \emph{reachability set of $X$}, and if $X$ is minimal, then $X = \delta^+(R_X)$.

An $(S,T)$-cut $X$ is an \emph{important $(S,T)$-separator} if there is no other $(S,T)$-cut $X'$ such that $|X'|\leq|X|$ and $R_X \subseteq R_{X'}$. (Observe that if $X$ is an important $(S,T)$-separator, then it is a minimal $(S,T)$-cut.)
Let $\bbX_k(S,T)$ denote the collection of all important $(S,T)$-separators in $G$ of size at most $k$. The following observation follows immediately from the definition of important separators.
\begin{observation}
    \label{obs:impsep-cover}
    Let $Y$ be any $(S,T)$-cut of size at most $k$ in $G$.
    Then, there is an important $(S,T)$-separator $X \in \bbX_k(S,T)$ such that $|X| \leq |Y|$ and $R_Y \subseteq R_X$.
\end{observation}

We have the following result on the collection ${\mathbb X}_k(S,T)$.

\begin{lemma}[Theorem 8.36~\cite{PCBook}]
    \label{thm:impsep}
    The cardinality of ${\mathbb X}_k(S,T)$ is upper bounded by $4^k$. Furthermore, ${\mathbb X}_k(S,T)$ can be computed in time $\calO(|\bbX_k(S,T)|k^2(m+n))$.
\end{lemma}

In this paper, we describe the construction of a $(\lambda,k)$-FTRS with respect to edge failures only, because any vertex failure can be modeled by an edge failure: Split every vertex $v$ into an edge $(v_{in},v_{out})$, where the incoming and outgoing edges
of $v$ are respectively directed into $v_{in}$ and directed out of $v_{out}$~\cite{BaswanaSICOMP18}.

\section{A simple algorithm for $(\lambda,k)$-FTRS}
\label{sec:k-ftss}

Let us first consider the case where $\lambda=1$, i.e.~the construction of a $k$-FTRS. Let $v \in V(G) \setminus \{s\}$ be an arbitrary vertex, and let $\bbX(v)$ denote the collection of all important $(s,v)$-separators in $G$ of size at most $k+1$.
By Lemma~\ref{thm:impsep}, there are at most $(k+1)4^{k+1}$ edges in the union of all such important separators. We have the following claim.

\begin{lemma}
    \label{lem:one-edge}
 Let $e \in \delta^-(v) \setminus (\bigcup_{X \in \bbX(v)} X)$. Then, $G-e$ is a $k$-FTRS of $G$.
\end{lemma}
\begin{proof} Suppose not, and consider a set $F$ of at most $k$ edges and a vertex $w \in V(G)$ such that $w$ is unreachable from $s$ in $(G-e)-F$, but there is a path from $s$ to $w$ in $G-F$. Since  $e \in \delta^-(v)$, it follows that $v$ is unreachable from $s$ in $(G-e)-F$, but there is a path from $s$ to $v$ in $G-F$. 
Thus, $Y = F \cup \{e\}$ is an $(s,v)$-cut in $G$. Since $G$ and $G-e$ only differ on $e$, this implies that $e$ lies on every path from $s$ to $v$ in $G-F$.
We may conclude the following: (a) $Y$ is an $(s,v)$-cut in $G$ of size at most $k+1$; (b) the vertex $u:=\mathrm{tail}(e)$ belongs to the reachability set of $Y$, i.e.~$u \in R_Y$.

Now, consider the $(s,v)$-cut $Y$ in the graph $G$ and the collection $\bbX(v)$. Since $e \notin \bigcup_{X \in \bbX(v)} X$, i.e.~$e \notin X$ for any $X \in \bbX(v)$, and $e = (u,v)$, it follows that $u \notin R_X$ for any $X \in \bbX(v)$. However, $u \in R_Y$ and $Y$ is $(s,v)$-cut of cardinality $k+1$ in $G$. 
This is a contradiction to Observation~\ref{obs:impsep-cover}.
Hence, $G-e$ is a $k$-FTRS of $G$.
\end{proof}

The above lemma can be turned into an iterative algorithm that gradually bounds the in-degree of each vertex in the graph. 
Let $\alpha = (k+1)4^{k+1}$ denote the upper-bound on $|\delta^-(v) \cap (\bigcup_{X \in \bbX(v)} X)|$ for any vertex $v \in V(G) \setminus \{s\}$.

\begin{algorithm}\label{algo}
\caption{An algorithm to compute a $k$-FTRS of a digraph $G$ with a source vertex $s$.}
\begin{algorithmic}[1]
    \Procedure {FTRS}{$G,s,k$}
    \State Delete all incoming edges of the source vertex $s$ in $G$.
    \While{there exists $v \in V(G) \setminus \{s\} \text{ such that } |\delta^-(v)| > \alpha$}
    \State Compute $\bbX(v)$ via Lemma~\ref{thm:impsep}.
    \State Pick an edge $e \in \delta^-(v) \setminus (\bigcup_{X \in \bbX(v)} X)$ and delete it.
    \EndWhile
    \EndProcedure
\end{algorithmic}    
\end{algorithm}

The correctness of the above algorithm follows from Lemma~\ref{lem:one-edge}. Furthermore, it is clear that this algorithm terminates once the in-degree of every vertex is upper bounded by $\alpha$. Hence, it runs in time $\calO(4^kk^2 m(m+n))$. This gives us the following corollary.
\begin{corollary}\label{cor}
    There is a $k$-FTRS where each vertex has in-degree at most $(k + 1) 4^{k + 1}$.
    Further, this graph can be constructed in time $\calO(4^kk^2(m+n)m)$.
\end{corollary}

The following simple lemma extends the above construction (and also the construction of Baswana et al.~\cite{BaswanaSICOMP18}) to any value of $\lambda$.
\begin{lemma}\label{lem:lambda}
    Let $H$ be a $(k+\lambda-1)$-FTRS of $G$. Then, $H$ is also a $(k,\lambda)$-FTRS of $G$.
\end{lemma}
\begin{proof}
    Suppose not. Consider a vertex $v \in V(G) \setminus \{s\}$
    and a set of $F$ of at most $k$ edges such that the following holds: there are $\lambda$ edge-disjoint paths from $s$ to $v$ in $G-F$, but there is no such collection of paths in $H-F$. Consider a minimum $(s,v)$-cut $X$ in $H-F$, and observe that $|X| \leq \lambda - 1$. It follows that $Y = F \cup X$ is an $(s,v)$-cut in $H$ of size at most $k+\lambda - 1$. However, $Y$ is not an $(s,v)$-cut in $G$.  This is a contradiction to the fact that $H$ is a $(k+\lambda-1)$-FTRS of $G$. Therefore, $H$ must be a $(\lambda,k)$-FTRS of $G$.
\end{proof}

The proof of Theorem~\ref{thm:k-ftss} follows from Corollary~\ref{cor} and Lemma~\ref{lem:lambda}.


\bibliographystyle{plain}
\bibliography{references}

\end{document}